\documentclass[12pt]{article}
\usepackage{amsmath}
\usepackage{amssymb}
\usepackage{epsfig}
\usepackage{textcomp}
\usepackage{amsthm}
\usepackage{footnote}
\usepackage{fullpage, bbm}
\usepackage[rgb]{xcolor}

\usepackage{verbatim, graphicx}
\usepackage{subfigure}
\usepackage{tikz}
\usetikzlibrary{calc,decorations.pathreplacing}
\usetikzlibrary{arrows,shapes}
\usetikzlibrary{calc,decorations.pathreplacing}
\usepackage{rotating}
\usepackage{marginnote}
\usepackage[marginparwidth=2cm]{geometry}

\newcommand{\cC}{\mathcal{C}}
\newcommand{\cH}{\mathcal{H}}
\newcommand{\cP}{\mathcal{P}}
\newcommand{\cQ}{\mathcal{Q}}

\newcommand{\JM}{\mathcal{M}}
\newcommand{\Pe}{\mathcal{P}^{ex}}
\newcommand{\out}{O}
\newcommand{\joint}{\mathcal{M}}
\newcommand{\basic}{X}
\newcommand{\mpout}{\text{Out}}
\newcommand{\measprot}{\text{MP}}

\newcommand{\la}{\langle}
\newcommand{\ra}{\rangle}

\definecolor{mygreen}{RGB}{5, 96, 5}

\newtheorem{theo}{Theorem}%[subsection]
\newtheorem{thm}[theo]{Theorem}
 
\newtheorem{lemma}[theo]{Lemma}

\newtheorem{defn}[theo]{Definition}
\theoremstyle{definition} %%% allows plain font for the following environments

\begin{document}
\title {Macroscopic non-contextuality as a principle for Almost Quantum Correlations}

\author{Joe Henson,  Ana Bel\'en Sainz\\[0.5em] 
{\it\small H.H. Wills Physics Laboratory, University of Bristol, Tyndall Avenue, Bristol, BS8 1TL, U.K.}}
\date{13 April 2015}

\maketitle
\begin{abstract}
Quantum mechanics allows only certain sets of experimental results (or ``probabilistic models'') for Bell-type quantum non-locality experiments. A derivation of this set from simple physical or information theoretic principles would represent an important step forward in our understanding of quantum mechanics, and this problem has been intensely investigated in recent years. ``Macroscopic locality,'' , which requires the recovery of locality in the limit of large numbers of trials, is one of several principles discussed in the literature that place a bound on the set of quantum probabilistic models. 

A similar question can also be asked about probabilistic models for the more general class of quantum contextuality experiments. Here, we extend the Macroscopic Locality principle to this more general setting, using the hypergraph approach of Ac\'in, Fritz, Leverrier and Sainz [Comm. Math. Phys. 334(2), 533-628 (2015)], which provides a framework to study both phenomena of nonlocality and contextuality in a unified manner. We find that the set of probabilistic models allowed by our Macroscopic Non-Contextuality principle is equivalent to an important and previously studied set in this formalism, which is slightly larger than the quantum set. In the particular case of Bell Scenarios, this set is equivalent to the set of ``Almost Quantum'' models, which is of particular interest since the latter was recently shown to satisfy all but one of the principles that have been proposed to bound quantum probabilistic models, without being implied by any of them (or even their conjunction). Our condition is the first characterisation of the almost quantum set from a simple physical principle.
\end{abstract}
\vskip 1cm
%=====================================================

\section*{Introduction}

Nonlocality \cite{Bell} and contextuality \cite{KS} \cite{RevBell, exp1, exp2, contexp0, contexp1, contexp2} are arguably the two phenomena that most starkly reveal the difference between quantum and classical mechanics  \cite{RevBell, LSW, conttheo1}. With regard to the first of these, Bell showed that quantum mechanics makes predictions for the strength of correlations between spacelike separated measurements that are incompatible with Bell's local causality condition, a seemingly natural formalisation of the idea that there is no superluminal causal influence \cite{Bell}. For the second, the Kochen-Specker theorem \cite{KS} states that quantum mechanics is incompatible with the assumption that measurement outcomes are determined by physical properties that do not depend on the measurement context. Quantum theory successfully explains both phenomena, but it predicts that only some sets of experimental probabilities for nonlocality and contextuality experiments can be attained, and it remains an open challenge to characterise this set of physically attainable results with simple, natural physical or information-theoretic principles.  The search for this deeper understanding of quantum nonlocality and contextuality is motivated by the possibility of reformulating and/or generalising quantum theory (\textit{e.g.}~for the purposes of formulating a theory of quantum gravity) as well as finding new ways to prove results in quantum information theory directly from simple principles, without having to invoke the whole structure of quantum theory.

The principle that there are no superluminal signals is not enough to characterise quantum correlations for nonlocality experiments in this sense \cite{tsirelson, PR}. Hence, stronger principles are needed. Some proposals are non-trivial communication complexity \cite{vd}, information causality (IC) \cite{IC}, local orthogonality \cite{FSA} and macroscopic locality \cite{ML}. In this work we focus on the macroscopic locality (ML) principle.  Essentially ML states that, for a certain macroscopic extension of a Bell experiment, Bell's local causality will hold, or in other words quantum nonlocality will no longer be detectable in this macroscopic limit.  

The problem of characterising quantum correlations in Bell scenarios from basic principles is however far from being solved. All the principles proposed so far, except IC, have been shown to be satisfied by some supra-quantum correlations (and the same is suspected to be true of IC) \cite{AQ}. Indeed, there exists a set of correlations called ``almost quantum'' that is slightly larger than the quantum set and yet satisfies these principles, presenting a curious barrier to a full characterisation of the quantum set \cite{AQ}. However, even a characterisation of the almost quantum set from basic principles is still missing.

Moving on to contextuality scenarios, the problem of characterising quantum models from basic principles has not been so intensely studied. The ``almost quantum'' set of non-local correlations generalises in this case to a set called $\mathcal{Q}_1$, which is strictly larger than the quantum set.  The most relevant proposals to describe the latter are the \textit{Exclusivity} principle \cite{cab2} and \textit{Consistent Exclusivity} (CE), which (when defined as in definition 7.1.1 of \cite{AFLS}) impose the same constraints (compare CE to the definition of the E principle in e.g. \cite{cab2}). This E/CE principle has been applied in many different ways to contextuality scenarios \cite{AFLS, cab2, cab1}, but they were never strong enough to single out quantum models in the sense of \cite{AFLS}. Moreover, several of the most powerful results rely on auxillary assumptions, some more simple and physically compelling than others.  For instance, when assuming both CE and that all quantum models are \textit{inside} the physically allowed set of models, it can be shown that all models violating $\mathcal{Q}_1$ are \textit{outside} the physically allowed set of models \cite{cab1, AFLS}.

In this work we propose a generalisation of ML to arbitrary contextuality scenarios, which we call \textit{macroscopic non-contextuality} (MNC). We use the hypergraph approach to nonlocality and contextuality developed in \cite{AFLS} to represent such scenarios, which we briefly review below. We find that MNC characterises the particular set $\mathcal{Q}_1$ of probabilistic models, which includes supraquantum models. For Bell scenarios, this strengthens the original ML principle, because the set $\mathcal{Q}_1$ is equivalent to the set of almost quantum correlations \cite{AQ} in that case.  Thus, we provide the first characterisation of almost quantum correlations from basic physical principles. 

In section \ref{se:contsce} below, the hypergraph approach to contextuality is reviewed and the relevant sets of probabilistic models (quantum, $\mathcal{Q}_1$ and non-contextual) are defined.  In section \ref{se:MNC} macroscopic noncontextuality is defined in analogy to macroscopic locality, and shown to be equivalent to $\mathcal{Q}_1$.  A discussion of the physical motivation of the principle and some other details follow in section \ref{se:discussion}.

\section{Contextuality scenarios}\label{se:contsce}

In this paper we represent general contextuality scenarios, including Bell scenarios, as in the hypergraph approach to contextuality of \cite{AFLS}. This section provides a brief review of the notation as well as the sets of correlations which are relevant for our result. For more details on the formalism, the reader can consult \cite{AFLS}.

A contextuality scenario \cite{AFLS} is defined as a hypergraph $H = (V, E)$ whose vertices $v \in V$ correspond to the events in the scenario. Each \textit{event} represents an outcome obtained from a device after it receives some input or ``measurement choice''. The hyperedges $e \in E$ are sets of events representing all the possible outcomes given a particular measurement choice. The hypergraph approach assumes that every such measurement set is complete, in the sense that if the measurement corresponding to $e$ is performed, exactly one of the outcomes corresponding to $v \in e$ is obtained. Note that measurement sets may have non-trivial intersection; when an event appears in more than one hyperedge, this represents the idea that the two different operational outcomes should be thought of as equivalent, in a sense that will be specified further below.

A \textit{probabilistic model} on a contextuality scenario is an assignment of a number to each of the events, $p\,:\, V \, \to \, [0,1]$, which denotes the probability with which that event occurs when a measurement $e \ni v$ is performed. By defining probabilistic models in this way (rather than by a function $p_e(v)$ depending on the measurement $e$ performed), we are assuming that in the set of experimental protocols that we are interested in, the probability for a given outcome is independent of the measurement that is performed.\footnotemark~Because the measurements are complete, every probabilistic model $p$ over the contextuality scenario $H$ satisfies the normalisation condition $\sum_{v \in e} p(v) = 1$ for every $e \in E$.

\footnotetext{In standard discussions of quantum contextuality, ``the set of experimental protocols that we are interested in'' means carrying out a fixed set of measurements on a quantum system.  In this case, two outcomes always have the same probabilities if they correspond to the same measurement operator acting on the same Hilbert space. When discussing contextuality more generally, it is often (explicitly or implicitly) assumed that some naturally defined set of experiments will still be available, with respect to which outcomes can be identified in a similar way; in some cases this can be justified by appeal to general principles, especially lack of signalling between parties.}

Bell scenarios (see Ap. \ref{se:corr-meas}) are naturally incorporated in the hypergraph approach as a type of product of several contextuality scenarios, one for each local party. Specifically, in an $(n,m,d)$ Bell scenario the ``global'' events $v \in V$ can be asscoiated with a list of ``local'' outcomes for each party: $v = (a_1 \ldots a_n | x_1 \ldots x_n)$. The hyperedge set $E$ however does not have such a simple representation: it includes \textit{simultaneous measurements} as well as \textit{correlated measurements}, also denoted as \textit{branching measurements} \cite{joe-hist} or \textit{one-way LOCC measurements} \cite{LOCCandreas} (see Ap. \ref{se:corr-meas} for a fuller explanation). 

In what follows we revisit the definitions of \textit{classical}, \textit{quantum} and $\mathcal{Q}_1$ probabilistic models. For other interesting sets of models we refer the reader to \cite{AFLS}. 

\begin{defn}\label{qmdef} \textbf{Quantum models} [\cite{AFLS}, 5.1.1]\\
Let $H$ be a contextuality scenario. An assignment of probabilities $p: V(H)\to [0,1]$ is a \emph{quantum model} if there exist a Hilbert space $\cH$, a quantum state $\rho\in\mathcal{B}_{+,1}(\mathcal{H})$ and a projection operator $P_v\in\mathcal{B}(\mathcal{H})$ associated to every $v\in V$ which constitute projective measurements in the sense that
\begin{equation}
\label{qmeas}
\sum_{v\in e} P_v = \mathbbm{1}_{\mathcal{H}} \quad\forall e\in E(H) ,
\end{equation}
and reproduce the given probabilities,
\begin{equation}
\label{qrep}
p(v) = \mathrm{tr}\left( \rho P_v \right) \quad\forall v\in V(H) .
\end{equation}
The set of all quantum models is the \emph{quantum set} $\mathcal{Q}(H)$.
\end{defn}

Later some comments will be made on the meaning and consequences of generalising this definition to POVMs. 

Ac\'in, Fritz, Leverrier and Sainz \cite{AFLS} prove that, in the case of Bell scenarios, Def.~\ref{qmdef} accords with the usual definition of quantum correlations (see Def.~\ref{qcorrvan} in Ap. \ref{se:corr-meas}), meaning that each global measurement represented by the projectors $\{P_v\}_{v \in e}$ can consistently be expressed as a product of local projectors, one for each party, such that the projectors for different parties commute (and sum up to the identity). For instance, in the bipartite case $P_{ab | xy} = P_{a|x} P_{b|y}$, where $[P_{a|x}, P_{b|y}]=0$ for all $a,b,x,y$ and $\sum_{a} P_{a|x} = \mathbbm{1}_{\mathcal{H}}$ (similarly $\sum_{b} P_{b|y} = \mathbbm{1}_{\mathcal{H}}$). 
 
The following set of correlations will be important in the following argument.

\begin{defn}\label{q1mdef} $\cQ_1$ \textbf{models} [\cite{AFLS}, 6.1.2]\\
Let $H$ be a contextuality scenario. An assignment of probabilities $p: V(H)\to [0,1]$ is a $\cQ_1$ \emph{model} if there exists a ``$\cQ_1$ certificate'':  a p.s.d.~matrix ranging over all $v \in V(H)$, with a special column labelled $1$, such that for all $e \in E(H)$,
\begin{enumerate}
\item\label{q1a} $\sum_{u\in e} M_{uv} = M_{1v}$ and $\sum_v M_{1v} = M_{11}$ ;
\item\label{q1b}  $(u,v \in e$ and $ u \neq v) \, \Rightarrow M_{uv} = 0$;
\item\label{q1c}  $M_{vv} = p(v)$;
\end{enumerate}
The set of all these models is denoted $\mathcal{Q}_1(H)$.
\end{defn}

This set $\mathcal{Q}_1$ arises in \cite{AFLS} as the first level of a hierarchy of relaxations that converges to the quantum set. In addition, it is shown in Corollary 6.4.2 of \cite{AFLS} that when the contextuality scenario is a Bell scenario, the set $\mathcal{Q}_1$ coincides with the Almost Quantum set of correlations \cite{AQ}. 

An equivalent characterisation of $\mathcal{Q}_1$ models is useful in the main proof of section \ref{se:MNC}: 
\begin{lemma}\label{q1mdef2}
Given a scenario $H$, a matrix $M$ is a $\cQ_1$ certificate for a given behaviour $P$ iff it is a p.s.d.~matrix ranging over all $v \in V(H)$, and a special column labelled $1$, such that
\begin{enumerate}
\item\label{q12a} $\sum_{u\in e} M_{uv} = P(v)$ for all $u\in V(H)$ ;
\item\label{q12b} $(u,v \in e$ and $ u \neq v) \, \Rightarrow M_{uv} = 0$;
\item\label{q12c} $M_{vv} = P(v)$;
\item\label{q12d} $M_{1v} = P(v)$ and $M_{11} = 1$;
\end{enumerate}
\end{lemma}
\begin{proof}
Condition (\ref{q1b}) of Def.~\ref{q1mdef} is equivalent to (\ref{q12b}) of Lemma~\ref{q1mdef2}, and condition (\ref{q1c}) of Def.~\ref{q1mdef} to (\ref{q12c}) of Lemma~\ref{q1mdef2}.
Conditions (\ref{q1a}), (\ref{q1b}) and (\ref{q1c}) of Def.~\ref{q1mdef} easily imply (\ref{q12d}) of Lemma~\ref{q1mdef2}.  Assuming (\ref{q12d}) of Def.~\ref{q1mdef}, (\ref{q1a}) of Def.~\ref{q1mdef} is equivalent to (\ref{q12a}) of Lemma~\ref{q1mdef2}.
\end{proof}

Finally, classical probabilistic models are defined as follows:
\begin{defn}\label{classical} \textbf{Classical models} [\cite{AFLS}, 4.1.1]\\
Let $H$ be a contextuality scenario. An assignment of probabilities $p: V(H)\to [0,1]$ is a \emph{classical model} if it can be written as 
\begin{equation}\label{eq:clas}
p(v) = \sum_\lambda q_\lambda p_\lambda(v),
\end{equation}
where the weights $q_\lambda$ satisfy $\sum_\lambda q_\lambda = 1$, and $p_\lambda$ are deterministic probabilistic models, that is normalised models such that $p_\lambda(v)=\{0,1\} \quad \forall \, v,\,\lambda$.

The set of all these models is denoted $\cC(H)$.
\end{defn}

Expressed in this language, the most famous result in this field, the Kochen-Specker theorem \cite{KS}, is that there exist scenarios that admit quantum models but no classical models, implying that there exist scenarios $H$ such that set $\cC(H) \subsetneq \cQ(H)$. This phenomenon is referred to as \textit{contextuality}, and the set of classical models is also referred to as the set of \textit{noncontextual models}.

Besides Bell scenarios, another special kind of scenario will be of particular relevance below when we come to discuss macroscopic versions of microscopic scenarios.  They are sometimes called ``marginal scenarios'' \cite{AB} or ``joint measurement scenarios''.  Here, we imagine some list of constituent experiments labelled $m \in \basic = \{1,...,k\}$, each with an outcome in the set $\out=\{1,...,d\}$.  Some subsets of the constituent experiments are ``jointly measurable''\footnote{In the sense that there exists a physically implementable protocol to measure them at the same time.}, and these subsets of $\basic$ are collected in the set $\joint \subset 2^\basic$ ($2^\basic$ is the set of all subsets of $\basic$).  For each $C \in \joint$, experimental probabilities are then assigned to each specification of a value for each of the constituent measurements: $\Pe_C(\{ a_m \}_{m\in C})$ where $a_m \in \out$.  Marginal scenarios can be represented in the hypergraph approach to contextuality, as explained in appendix \ref{a:joint_measurements}.  As is also explained in that appendix, for this type of scenario the definition of classical models given above can be rewritten with (\ref{eq:clas}) becoming

\begin{equation}\label{e:nc_joint}
\Pe_C(\{a_m \}_{m\in C}) = 
\sum_{m \in \basic \backslash C}  \cP_{\text{NC}}( \{ a_m \}_{m \in \basic} ) ,
\end{equation}
where $\Pe_C(\{a_m \}_{m\in C})$ is the experimental probability of obtaining outcomes $\{a_m \}_{m\in C}$ given that the joint measurement $C$ was performed, and where $\backslash$ is set difference. 
In this form, the following interpretation of non-contextuality for marginal scenarios is brought out: the probabilities are such that the results of the constituent experiments are consistent with an outcome for every observable having been predetermined before the measurement is performed, and the experiment ``merely revealing'' the results for the ones that are measured.  It should be noted that unlike the most general scenarios that can be represented in the hypergraph approach, all the (non-empty) joint measurement scenarios have a non-empty set of classical models.

\section{Macroscopic Non-Contextuality}\label{se:MNC}

In \cite{ML} Navascu\'es and Wunderlich identify an interesting property of quantum correlations which they termed \textit{macroscopic locality} (ML), and proposed that this be thought of as a simple physical principle to bound the set of correlations. Essentially, macroscopic locality requires that a certain ``macroscopic limit'' of a Bell-type experiment has a local explanation in the sense of Bell.  In \cite{ML} the principle was applied to bipartite Bell scenarios, and shown to be equivalent to the first level of the NPA hierarchy \cite{NPA0, NPA}. Hence, the set of correlations which satisfies the principle is strictly larger than the set of quantum correlations. In this section we extend this kind of reasoning to general contextuality scenarios in the hypergraph approach, including multipartite Bell scenarios as special cases, as described above.  We prove that the set of probabilistic models satisfying this principle is, again, strictly larger that the quantum set $\cQ$, but that it is stronger than Navascues and Wunderlich's ML when specialised to Bell scenarios. 

Consider a physical system $s$ and a set of measurements $E$, from which we choose one to perform on $s$. As reviewed above, in the hypergraph approach to contextuality such a scenario is represented by a hypergraph $H=(V,E)$; the (normalised) probability $p(v)$ of obtaining an outcome $v \in V$ given that a measurement $e \ni v$ is performed, for all outcomes, defines a probabilistic model on $H$. An experiment of this type is depicted in Fig.~\ref{micro-exp}, and we refer to it as \textit{microscopic experiment}. Now we want to define a macroscopic version of such an experiment, which we call its ``macroscopic extension''. Suppose now that the source produces $N$ independent copies of this system $s$, and that these $N$ systems reach the measurement device (see Fig.~\ref{macro-exp}). Now we assume that we are no longer able to distinguish individual outcomes, but only the fraction of instances (or ``intensity'') of each outcome $v$ given a measurement $e$.  The experimental results for a particular measurement in the macroscopic experiment are thus described by a probability distribution $\cP_e(\{I^v\}_{v \in e})$ where $I^v$ denotes the intensity for outcome $v$. 
This can be described as a joint measurement scenario in which the constituent experiments are the measurements of the intensity $I^v$ for each $v$.\footnote{Although here we must allow continuous values for the intensities, the generalisation does not change anything important for our purposes.}  The probabilities for the macroscopic extension are determined by the microscopic probabilistic model $p(v)$, in a way that we will make explicit below.  

\begin{figure}
\begin{center}
\begin{tikzpicture}
\draw[thick, color=gray!70!black, ->] (0,0) -- (2,0);
\draw[thick, color=gray!70!black, ->] (4,0) -- (6,0);
\draw[thick, color=gray!70!black, ->] (4,0) to [out=0,in=180] (6,1.25);
\draw[thick, color=gray!70!black, ->] (4,0) to [out=0,in=180] (6,-1.25);
\node[draw=gray!70!black,shape=circle,shading=ball, ball color=gray!50!white,scale=1] at (0,0) {} ;
%\node[draw=gray!60!black,shape=rectangle,shading=ball, ball color=gray,scale=3] at (2,0) {} ;
\draw[draw, thick, color=gray!90!black,rounded corners, inner color=white,outer color=gray!50!white] (2,-1) rectangle (4,1) ;
\draw[draw, thick, color=gray!90!black, inner color=white,outer color=gray!50!white] (6,-0.25) arc (270:90:-0.4cm and 0.25cm) -- (6,-0.25);
\draw[draw, thick, color=gray!90!black, inner color=white,outer color=gray!50!white] (6,1) arc (270:90:-0.4cm and 0.25cm) -- (6,1);
\draw[draw, thick, color=gray!90!black, inner color=white,outer color=gray!50!white] (6,-1.5) arc (270:90:-0.4cm and 0.25cm) -- (6,-1.5);
\node at (6.7, 1.25) {D$_1$};
\node at (6.7, 0) {D$_2$};
\node at (6.7, -1.25) {D$_{|e|}$};
\node at (3, -1.5) {M};
\node at (0, -0.5) {S};
\node at (1, 0.25) {$s$};
\end{tikzpicture}
\end{center}
\caption{Microscopic experiment. A source S prepares a system $s$, which is sent to the measurement device M. There, an interaction between the measurement apparatus and the system sends the system towards one of a set of detectors, where its presence can be observed as a ``detector click''.  The clicking of detector D$_k$ corresponds to obtaining outcome $k$.}
\label{micro-exp}
\end{figure}
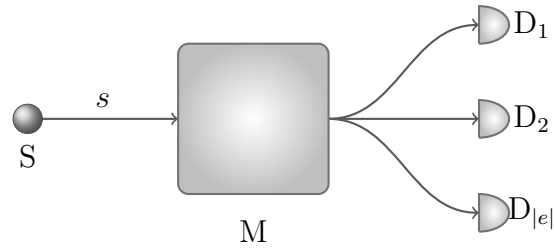

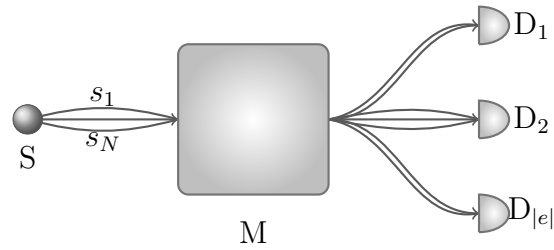
\begin{figure}
\begin{center}
\begin{tikzpicture}
\draw[thick, color=gray!70!black, ->] (0,0) -- (2,0);
\draw[thick, color=gray!70!black] (0,0) to [out=-15, in=195] (2,0);
\draw[thick, color=gray!70!black] (0,0) to [out=15, in=165] (2,0);
\draw[thick, color=gray!70!black, ->] (4,0) -- (6,0);
\draw[thick, color=gray!70!black] (4,0) to [out=-10, in=195] (6,0);
\draw[thick, color=gray!70!black] (4,0) to [out=10, in=165] (6,0);
\draw[thick, color=gray!70!black, ->] (4,0) to [out=0,in=180] (6,1.25);
\draw[thick, color=gray!70!black] (4,0) to [out=10,in=170] (6,1.25);
\draw[thick, color=gray!70!black, ->] (4,0) to [out=0,in=180] (6,-1.25);
\draw[thick, color=gray!70!black] (4,0) to [out=-10,in=190] (6,-1.25);
\node[draw=gray!70!black,shape=circle,shading=ball, ball color=gray!50!white,scale=1] at (0,0) {} ;
%\node[draw=gray!60!black,shape=rectangle,shading=ball, ball color=gray,scale=3] at (2,0) {} ;
\draw[draw, thick, color=gray!90!black,rounded corners, inner color=white,outer color=gray!50!white] (2,-1) rectangle (4,1) ;
\draw[draw, thick, color=gray!90!black, inner color=white,outer color=gray!50!white] (6,-0.25) arc (270:90:-0.4cm and 0.25cm) -- (6,-0.25);
\draw[draw, thick, color=gray!90!black, inner color=white,outer color=gray!50!white] (6,1) arc (270:90:-0.4cm and 0.25cm) -- (6,1);
\draw[draw, thick, color=gray!90!black, inner color=white,outer color=gray!50!white] (6,-1.5) arc (270:90:-0.4cm and 0.25cm) -- (6,-1.5);
\node at (6.7, 1.25) {D$_1$};
\node at (6.7, 0) {D$_2$};
\node at (6.7, -1.25) {D$_{|e|}$};
\node at (3, -1.5) {M};
\node at (0, -0.5) {S};
\node at (1, 0.3) {$s_1$};
\node at (1, -0.3) {$s_N$};
\end{tikzpicture}
\end{center}
\caption{Macroscopic experiment. A source S prepares $N$ independent copies of a system $s$, which are sent to the measurement device M. There, for each system (and independently for each system), an interaction between the measurement apparatus and the system sends the system towards one of a set of detectors,  However, in this case, rather than a single click, there is a distribution of `clicks' over the detectors according to the probabilities for each outcome in the microscopic experiment. Hence, the `output' of this macroscopic experiment is the collection of intensities $I^v_e$ registered at the detectors.}
\label{macro-exp}
\end{figure}

Generalising ML \cite{ML}, our principle will be that in the limit of large $N$ there exists a non-contextual model for this experiment: the probabilities are such that the intensities for \textit{all} of the outputs $v$ could have been predetermined before the measurement is performed, and the experiment ``merely reveals'' the intensities that are measured. 

\begin{defn}\textbf{Macroscopic Non-Contextuality } \textsl{(MNC)}\\
The probabilistic model $p(v)$ obeys macroscopic non-contextuality if, in the limit $N \rightarrow \infty$, there exists a probability distribution $\cP_{\text{NC}}$ over a set of intensities $\{ I^v \}_{v \in V(H)}$, such that the experimental probabilities for the macroscopic extension of $p(v)$, $\cP_e(\{I^v\}_{v \in e})$, can be obtained as marginals from $\cP_{\text{NC}}$: 
\begin{equation}\label{e:nc_prob_dist}
\cP_e(\{I^v\}_{v \in e}) = 
\int \Bigl ( \prod_{v \in V(H) \backslash e} dI^v \Bigr ) \; \cP_{\text{NC}}( \{ I^v \}_{v \in V(H)} ) ,
\end{equation}
where $\backslash$ is set difference.
\end{defn}

Equation (\ref{e:nc_prob_dist}) is the analogue of Eq.~(\ref{e:nc_joint}) for the macroscopic experiment.  Note that no matter what the scenario $H$ is for the microscopic experiment (at least as long as it supports any probabilistic models at all), this condition can always be satisfied by some probability distributions.  The original $H$ may even constitute a proof of Kochen-Specker, but the corresponding macroscopic experiment is always represented by a marginal scenario, which is never of that type.

We will now investigate how to characterise the set of probabilistic models $p(v)$ that satisfy MNC.  It is useful to discuss this question in terms of random variables.  In general, the results of the macroscopic experiment are described by the probability distributions $\cP_e(\{I^v\}_{v \in e})$. 
With a slight abuse of notation, we denote by $I^v_e$ the random variable associated to variable $I^v$ in the distribution $\cP_e$.  This is done because the random variables derived from different distributions are distinct, and those corresponding to the same outcome would share the same symbol without the added subscript. Note that the random variables $I^u_e$  and $I^v_f$ for $e \neq f$ are defined from different distributions and so it is meaningless to ask about correlations between them.

A macroscopic experiment is defined from $N$ ``runs'' of the microscopic experiment.  Similarly to \cite{ML}, define $d^v_{i\,e}$ as a random variable that is $1$ if $v$ is obtained in the $i$th run of experiment $e$ and $0$ otherwise. The intensity of outcome $v$ given measurement $e$, $I^v_e$, is then proportional to $\sum_{i=1}^N d^v_{i\,e}$, and its deviation from the mean value is expressed as:
\begin{equation}
\label{e:def_I}
\bar{I}^v_e = \sum_{i=1}^N \frac{\bar{d}^v_{i\,e}}{\sqrt{N}} = \sum_{i=1}^N \frac{d^v_{i\,e}-p(v)}{\sqrt{N}},
\end{equation}
where the normalisation has been chosen to be $\sqrt{N}$ for reasons that will hopefully become clear below.

There are some constraints on the random variables $\bar{I}^v_e$ that follow from the consistency of the probabilistic model for the microscopic experiment. The first simply comes from the fact that the sum of the number of hits for all outcomes over all runs must be N, so that
\begin{equation}
\label{e:overall_norm}
\sum_{v \in e} \bar{I}^v_e=0 \quad \forall\,e.
\end{equation}

The second only holds in the limit.  The central limit theorem \cite{CLT} implies that, when $N \to \infty$, the probability distribution over the intensity fluctuations for each experiment converges to a multivariate Gaussian distribution.  In this case the covariance matrix $\gamma^e$ for the experiment $e$ will be given by the following, defined for all $u,v \in e$,
\begin{equation}
\label{e:central_lim}
\gamma^e_{uv} = \la \bar{I}^u_e \bar{I}^v_e \ra = \la \bar{d}^u_{1 \,e} \bar{d}^v_{1 \,e} \ra= \delta_{u v} p(v) - p(u)p(v).
\end{equation}
Note that the value of $\gamma^e_{uv}$ is the same for fixed $u$ and $v$, for any value of $e$.  This is because the marginal distribution $\Pe_e(I^u,I^v)$ for some measurement containing $u,v$ as outcomes is the same no matter what $e$ is (this in turn follows from the consistency of the probabilistic model for the microscopic experiment and the definition of the intensities).

So far these observations hold in general (\textit{i.e.}~without any constraint being placed on the microscopic model beyond consistency).  Now, if MNC holds, then in the limit $N \rightarrow \infty$ there exists a joint probability distribution over the set of intensitites for \textit{all} outcomes, such that the experimental distributions can be recovered as marginals as in (\ref{e:nc_prob_dist}).  In terms of random variables we can now define $I^v$, without any subscript denoting a measurement, from $\cP_{\text{NC}}(\{ I^v \}_{v \in V(H)} )$.  These $I^v$ are all derived from the same disriibution and so MNC implies that there must exist a bigger matrix $\gamma_{uv}$ defined for \textit{all} $u,v \in V(H)$ that has the properties of a covariance matrix for this probability distribution; in particular it is a positive semi-definite matrix.  Furthermore, from (\ref{e:nc_prob_dist}) this $\gamma_{uv}$ must reduce to (\ref{e:central_lim}) when $u,v$ are restricted to $e$.  A further constraint on $\gamma_{uv}$ implied by (\ref{e:nc_prob_dist}) and (\ref{e:overall_norm}) is that, even for $u$ not in the same measurement as $v$,
\begin{equation}
\sum_{u\in e} \gamma_{uv} = \la (\sum_{u\in e} \bar{I}^u) \bar{I}^v \ra = 0.
\end{equation}
Hence, the microscopic probabilistic models which are consistent with MNC may be characterised as follows:
\begin{defn}
\label{d:mnc}
A probabilistic model $p$ on scenario $H$ is macroscopically non-contextual if there exists a ``macroscopic non-contextuality certificate'': a p.s.d.~matrix $\gamma$ ranging over all $v \in V(H)$ such that
\begin{itemize}
\item $\sum_{u\in e} \gamma_{uv} = 0$;
\item $(u,v \in e \text{ and } u \neq v) \, \Rightarrow \gamma_{uv} = -p(u)p(v)$;
\item $\gamma_{vv} = p(v)-p(v)^2$;
\end{itemize}
\end{defn}

Our main result is that these microscopic probabilistic models are equivalent to the $\mathcal{Q}_1$ set in the hierarchy of pobabilistic models defined in the hypergraph approach \cite{AFLS}. 

\begin{thm}
A behaviour is macroscopically non-contextual iff it is in $\cQ_1$.
\end{thm}
\begin{proof}
The proof is very similar to the one in \cite{ML} for ML. From Schur's theorem \cite{Horn}, because $M_{11}=1>0$, the positivity of $M$ is equivalent to the positivity of $\gamma_{uv}= M_{uv} - M_{1v}M_{1u}=M_{uv} - p(u)p(v)$.

With this definition, it is easy to check that (\ref{q1mdef2}.\ref{q12a}) is equivalent to $\sum_{u\in e} \gamma_{uv} =0$, (\ref{q1mdef2}.\ref{q12b}) is equivalent to $(u,v \in e \text{ and } u \neq v) \, \Rightarrow \gamma_{uv} = -p(u)p(v)$ and (\ref{q1mdef2}.\ref{q12c}) is equivalent to $\gamma_{uv} = p(v)-p(u)p(v)$.  Given the probabilistic model, the values of $M_{1v}$ and $M_{11}$ are determined.  Thus one can derive a macroscopic non-contextuality certificate $\gamma$ given that there exists a $\cQ_1$ certificate, and \textit{vice-versa}.
\end{proof}

Since quantum models are included within the $\mathcal{Q}_1$ set \cite{AFLS}, Quantum theory satisfies MNC for any contextuality scenario.

\section{Discussion}\label{se:discussion}

In the particular case of Bell Scenarios, the set $\mathcal{Q}_1$ is equivalent to the set of almost quantum correlations (see \cite{AFLS} theorem 6.4.1, and similarly, \cite{AQ} lemma 3). Hence, while the original version of ML allows a strictly larger set of correlations than the almost quantum set \cite{AQ}, when Bell scenarios are defined as in the hypergraph approach a stronger version of the principle arises, which allows exactly the almost quantum correlations.  This is the first time that a simple physical principle has been shown to limit correlations to the almost quantum set\footnote{In \cite{joe-hist} the set SPJQM$_b$ is shown to be equivalent to the almost quantum set.  However, in \cite{joe-hist} the motivation is slightly different, looking for a natural and useful generalisation of quantum mechanics, rather than a derivation of properties of QM from principles of the type discussed in the present work.  It is difficult to call SPJQM$_b$ a simple physical principle in itself, since in \cite{joe-hist} the positive-semidefiniteness is ``added by hand'' rather than derived.  Nonetheless, it is suggestive and intriguing for both programs that the same set of correlations has been arrived at from these different starting points.}.

The essential differences that lead to this strengthening are the consideration of global outcomes rather than local ones, and the inclusion of the \textit{correlated measurements} in the definition.  However, the gedanken experiment used above to motivate MNC was not of the form of a physical Bell experiment and so some care is needed here when it comes to the motivation for applying the condition. These issues are discussed in detail in appendix \ref{se:corr-meas}. The question of how to motivate correlated measurements in general nonlocality scenarios however goes beyond the scope of this manuscript and is deferred to future work \cite{fuwo}.

This situation is somewhat similar to the difference between the sets SPJQM and SPJQM$_b$ in the histories approach to quantum non-locality \cite{joe-hist}. There, including correlated (or ``branching'') measurements allowed the authors to recover almost quantum correlations from a condition that otherwise has only been shown to imply the first level of the NPA hierarchy \cite{NPA0, NPA}.  Also, the almost quantum set is much more naturally defined on the hypergraph approach version of Bell scenarios than the first NPA set,  while the latter is the more natural condition when correlated measurements are not considered, and instead it is directly imposed that all mathematical objects associated to local outcomes are independent of the distant measurement settings.  This suggests that, in order to characterise the almost quantum set, it is necessary to bring in considerations of correlated measurements (although it is of course possible that a different way to motivate the same strengthening may be found).

A \textit{wiring} is a classical operation by which a new probabilistic model $p$ is constructed from a set of models $\{p_1, \ldots, p_r\}$. For example, in a tripartite Bell scenario a
wiring may consist of Alice communicating her outcome $a$ to Bob, who uses this outcome as his choice of measurement and obtains an outcome $b$. One can then define a new probabilistic model $p$ from the original tripartite one upon identifying Alice and Bob with a new joint party with joint measurement choice $x$ and joint outcome $b$. This type of classical operation can increase the violation of a Bell inequality \cite{Brun}. However, there is very strong motivation to assume that the set of probabilistic models that arises within a physical theory is closed under these classical operations \cite{Brun}. This is indeed the case for the set of MNC probabilistic models. When considering a general contextuality scenario the possible classical operations include choosing one measurement from many via a probability distribution, or in considering many devices (\textit{i.e.}~systems) ``in parallel'' as one larger device.  In this regard, it is shown in \cite{AFLS} that the set $\mathcal{Q}_1$ is both convex and closed under tensor products. However, for the particular case of Bell scenarios there exists a larger set of classical operations one could consider. This is studied in \cite{AQ}, where it is proven that the set of Almost quantum correlations is closed under post-selection, grouping of parties and composition. Hence, whenever a collection of microscopic probabilistic models satisfy MNC, the result of any such wiring operation among them will satisfy the principle as well. 

Another question that is suggested by the above result is how far the principle can be pushed.  One could argue that the most general measurement in quantum mechanics is given by a POVM, and so Def.~\ref{qmdef} should be generalised, substituting positive operators for projectors.  Will the MNC principle still be true for all ``quantum correlations'' when we allow this?  And if the principle fails in these cases, should that not cast doubt on the claim that the principle is a physically reasonable restriction, undermining the motivations discussed above?

In fact the principle \textit{does} fail in this case, but our view is that this only highlights how problematic it is to generalise Def.~\ref{qmdef} to POVMs. Indeed, one can take any general probabilistic model $p(v)$ on the contextuality scenario $H$ (that is, any model that satisfies the normalisation constraints), and define the positive operators $P(v):= p(v) \, \mathbbm{1}$, where $ \mathbbm{1}$ is the identity on $\cH$. These operators will always satisfy the conditions of Def.~\ref{qmdef} generalised to POVMs. Since $p(v)$ can be any probabilistic model, if we allow general POVMs rather than projective measurements then \textit{no} principle that places a non-trivial restriction on correlations will be respected.  Thus, this kind of ``quantum model'' is clearly pathological.  Furthermore, there is nothing special to quantum theory about this: one could apply an analogous generalisation to classical models with similar results. In this case, the analogue of POVMs would be to incorporate classical randomness (``noise'') into the measurements.  But this would allow the unmysterious form of contextuality in which identified outcomes do not in fact to correspond to the same property of the system in any meaningful sense. As already noted by Spekkens \textit{et al.} \cite{Spe, LSW}, the relation of POVMs to contextuality demands more careful consideration (see also \cite{Kunjwal:2014} for further considerations along these lines).

Finally, a note on the meaning of the $N \rightarrow \infty$ limit being used here is in order.  We have assumed that (\textit{a}) we cannot look for correlations between individual runs on the experiment but only between the proportions of outcomes over all runs, and (\textit{b}) that at large $N$ the experimenter has the ability to resolve the fluctuations described by the CLT but not the deviations from this due to the finiteness of $N$.  In effect the experimenter must have a resolution that can pick out fluctuations of order $\sqrt{N}$.  This is stronger than the resolution necessary to measure the mean values of the intensities but weaker than would be necessary to resolve the microscopic structure in the stronger sense of seeing finite $N$ effects.  It is very intriguing that quantum mechanics turns out to be non-contextual in this natural limit. 

\section{Conclusions}
In this work we have proposed a strengthening of Macroscopic Locality, called Macroscopic Non-Contextuality, as a new principle to bound Quantum models on general contextuality scenarios. We have used the hypergraph approach to nonlocality and contextuality to represent such scenarios, and proven that our principle is equivalent to the first level  ($\mathcal{Q}_1$) in the hierarchy of probabilistic models defined in the hypergraph approach.

In the hypergraph approach representation of Bell scenarios, the $\mathcal{Q}_1$ set corresponds to Almost Quantum correlations, hence our approach provides a natural characterisation of this set. The inclusion of one-way LOCC measurements as feasible actions in a Bell scenario seems to be the key ingredient that allows the strengthening of the original ML principle. One-way LOCC measurements are important in information theoretic tasks, such as local distinguishability of quantum states \cite{oneway1}. Some questions remain open surrounding the physical motivation for considering this type of correlated measurement in nonlocality and contextuality scenarios, which will be deferred to future work \cite{fuwo}.

There are related programs which also treat the problem of characterising the set $\mathcal{Q}_1$. The exclusivity principle characterises $\mathcal{Q}_1$ when it is assumed in addition that quantum models are all included in the physically allowed set of models (the inclusion of this assumption defines \textit{Extended Consistent Exclusivity} (ECE) in the nomenclature of the hypergraph approach \cite{AFLS}). One of the main differences between this approach and ours is that our approach does not need the extra assumption. The other main difference involves the application of the principle to Bell scenarios. To derive the bound from the ECE principle for a Bell scenario, it is necessary to assume that some non-Bell scenarios can be realised and must also obey the same set of assumptions; the proof does not go through if considerations are restricted to nonlocality scenarios alone. In contrast, the derivation of the $\mathcal{Q}_1$ bound from MNC for a particular scenario involves no considerations of other scenarios at all.  Thus, MNC may be used as a principle to characterise correlations in Bell scenarios solely, and successfully recovers the almost quantum set.

In view of the fact that the Almost Quantum set satisfies (or at least has not been shown to violate) all of the principles proposed so far \cite{AQ}, the most important outstanding question is how to formulate a principle that gets closer to quantum models. In the present work, similarly to the original ML paper \cite{ML}, we focus on experiments where only one of many possible measurements is performed on the system. However, another physically relevant experiment could be defined by applying sequences of measurements. The formulation of an MNC-like principle for such experimental scenarios is an open problem, whose solution we believe may shed light on the important question of how to distinguish quantum from almost quantum correlations from basic natural principles.

A short discussion on the Almost quantum set is still in order. Even though there exist several mathematical characterisations of this set, the relation of the Almost Quantum behaviours to physical theories is still unclear and deserving of further research. For instance, no toy theory is available that produces this set of behaviours. It is also not yet known how to express in a physical form a simple property of quantum mechanics that rules out the supra-quantum almost quantum correlations. We believe that the understanding of these open problems will shed light on the characterisation of quantum theory. 

\section*{Acknowledgements}

We thank Miguel Navascu\'es, Tony Short and Sandu Popescu for fruitful discussions and comments. This research was supported by EPSRC grant DIQIP and ERC AdG NLST. 

\appendix 

\section{Marginal Scenarios and non-contextuality in the hypergraph approach.}\label{a:joint_measurements}

There is a special kind of scenario called ``marginal scenarios'' (or ``joint measurement scenarios'') which often appear in discussions of contextuality.  These are the central structure of an alternative formalism for contextuality called the ``observable-based'' approach \cite{AB}, which can be shown to be essentially equivalent to the hypergraph approach in the sense that one can be derived from the other (although additional constraints must be added to the observable-based formalism to recover the hypergraph-based formalism).  These joint measurement scenarios appear in the main argument of this paper as the representation of macroscopic experiments.

In appendix D of \cite{AFLS} the relationship between the joint measurement scenarios and the hypergraph approach to contextuality is explained. In this appendix, we briefly reprise the relevant points made there, and make explicit the relationship between classical non-contextual models in the hypergraph approach and the natural condition for non-contextuality that we apply to marginal scenarios in the main text.  Here we will use a slightly less abstract notation than that in appendix D of \cite{AFLS}, but otherwise the terminology will be the same.

Consider some list of constituent or ``basic'' experiments, which we will call ``observables'', labelled $m \in \basic = \{1,...,k\}$, each of which has outcomes in the set $\out=\{1,...,d\}$ (here we have assumed that all observables are valued in the same set).  Some subsets of the constituent experiments are ``jointly measurable'' and these subsets of $\basic$ are collected in the set of ``measurement contexts'' $\joint \subset 2^\basic$.  Here we will only be interested in ``maximal'' measurement contexts which cannot be extended by adding more observables, and so we assume that, for any $C,C' \in \joint$, if $C \subseteq C'$ then $C=C'$. We also assume that for every $m\in \basic$ there exists a $C \in \joint$ that contains it.  As is common practice in such cases we represent the marginal scenario defined by $(\basic,\joint,\out)$ just by $X$ when the context makes it clear what is meant.  

Some auxillary definitions are necessary before we can define a hypergraph approach scenario $H[X]$ that is equivalent to this.  In the hypergraph approach each event $v \in V(H)$ represents a full (``global'') specification of the experimental outcome.  Here, given a set of jointly measurable observables $C \in \joint$ a possible outcome is specified by $\{ a_m \}_{m \in C}$ where $a_m \in \out$.  These outcomes have to be distinguished for every $C \in \joint$,  and therefore we set
\begin{equation}
V(H[X]) :=  \bigl\{ (C, \{ a_m \}_{m \in C} ) : C \in \joint, \{ a_m \}_{m \in C} \in O^{C} \bigr\}.
\end{equation}

That is, we have a disjoint set of ``global'' outcomes for every (maximal) measurement context.  The definition of the measurement set $E(H[X])$ requires more care.  As well as simply measuring all observables in a measurement context, we could measure each observable in turn, and choose which observable to measure next depending on the results obtained so far, until we had an outcome for each observable in one of the maximal sets of jointly measurable observables.  For such a measurement, the full list of alternative global outcomes is not just a list of all possible combinations of local outcomes for one fixed measurement context.  We will define these ``measurement protocols'' recursively.

Measuring an observable $A$ limits our options on what we can measure next. Given the first measured observable $m$, the remaining possibilities can be represented as an ``induced'' marginal scenario $(\basic\{m\},\joint\{m\},\out )$, for which 

\begin{align}
\basic\{m\} &:=  \bigl\{ m' : m\neq m', \exists C \in \joint  \text{ s.t. } \{m,m'\} \subseteq C \bigr\}, \\
\joint\{m\} &:=  \bigl\{ C \backslash \{m\} : C \in \joint \text{ and } m \in C \bigr\}. 
\end{align}

A measurement protocol $T(X)$ on the marginal scenario $X$ can now be defined in the following way:  $T=\emptyset$ if $X=\emptyset$ and otherwise  $T=(m,f)$ where $m\in\basic$ is an observable and $f: \out \rightarrow T(X\{m\})$ is a function from outcomes to measurement protocols on $X\{m\}$.  In words, we first choose an observable $m$ to measure, and then decide between all protocols for choosing subsequent observables to measure based on the outcome. We continue making these choices until we can no longer find a jointly measurable observable to add to the set that we have already measured.  The set of all possible outcomes for a protocol $T=(m,f)$ can be defined in a similar way by including the outcome of the observable in the recursive structure:
\begin{equation}
 \mpout(T) := \bigl \{ (m,a_m,\alpha') : a_m \in O, \alpha' \in \mpout(f(a_m)) \bigr \}.
\end{equation}
Using these recursion relations, each outcome $\alpha\in \mpout(T)$ uniquely specifies an event $v \in V(H[X])$.  That is, $v_\alpha=(C_\alpha, \{ a_m \}_{m \in C_\alpha})$ where, if $\alpha=(m,a,\alpha')$ as in the above equation, then $C_\alpha= \{m\} \cup C_{\alpha'}$ on $X$, and $\{ a_m \}_{m \in C_\alpha}$ is the set of outcomes associated to $\alpha$ in the obvious way.  Finally, the measurement sets are
\begin{equation}
E(H[X]) := \{e_T : T \in \measprot(X) \}
\end{equation}
where $ \measprot(X) $ is the set of all measurement protocols on $X$ and 
\begin{equation}
e_T := \bigl \{ (C_\alpha,\{ a_m \}_{m \in C_\alpha}) : \alpha \in \mpout(T) \}.
\end{equation}

In \cite{AFLS} the relationship between the most general consistent probability structures in the two formalisms (``empirical models'' on marginal scenarios and ``probabilistic models'' for the hypergraph approach) is established, and it is noted that there is a similar relationship between the definitions of classical noncontextual models and quantum models in the two cases as well.  Here it is necessary to make the former connection explicit:  non-contextuality for marginal scenarios is used to define macroscopic non-contextuality in the main text, while otherwise the hypergraph approach has been employed, and so this begs the question of the connection between the two.

Let us consider the possible deterministic probabilistic models on a scenario $H[X]$.   If the only measurements included in our considerations were the ones corresponding to maximal measurement contexts $C \in \joint$, then, because these measurements correspond to disjoint sets of outcomes, for any choice of an outcome for every $C$ there would be a deterministic probabilistic model that assigned probability 1 to those outcomes only.  However, this would allow the implied outcome for a particular observable to depend on the overall context $C$ in which it was measured.  The definition of $E(H[X])$ given above, including all measurement protocols, prevents this, as we explain in the following. 

Consider a pair of outcomes for a pair measurement contexts which imply different outcomes for some particular observable $m^* \in \basic$.  Formally, this pair is some $u=(C,\{ a_m \}_{m \in C})$, $v=(C',\{ a'_m \}_{m \in C'})\in V(H[X])$ such that there exists an observable $m^* \in \basic$ with $m^* \in C$ and $m^* \in C'$ but with $a_{m^*} \neq a'_{m^*}$.  Any such pair is contained in some measurement $e \in  E(H[X])$ (as can be seen by making $m^*$ the first measurement in some measurement protocol $T=(m^*,f)$ and chosing the function $f$ appropriately so that $u,v \in  \mpout(T)$).  Considering this, in a deterministic probabilistic model on $H[X]$, no such pair can both be assigned probability 1, or in other words, the probability of obtaining a particular value for an observable cannot depend on the context $C$.  It is not difficult to see that there are no other constraints on deterministic probabilistic models on $H[X]$.  Any such model is thus completely specified by assigning an output to every observable, $\{a_m\}_{m \in \basic}$.

In the light of this, consider the definition of a classical model given in section \ref{se:contsce}. For marginal scenarios, for an event $u=(C, \{\hat{a}_m \}_{m \in C})$, the definition of classical models (\ref{eq:clas}) now takes the form

\begin{equation}\label{e:nc_joint_deltas}
\Pe_C(\{ \hat{a}_m \}_{m \in C}) := P \bigl (  (C, \{ \hat{a}_m \}_{m \in C} ) \bigr) =
\sum_{C \in \joint}  \cP_{NC}( \{ a_m \}_{m \in \joint} ) q_{\{ a_m \}_{m \in \joint}}(\{\hat{a}_m \}_{m\in C}),
\end{equation}

where $a_m \in A_m$, $\cP_{NC}( \{ a_m \}_{m \in \joint} )$ is a probability distribution over $\{ a_m \}_{m \in \joint}$,  and
\begin{equation}
q_{\{ a_m \}_{m \in \JM}}(\{\hat{a}_m \}_{m\in C}) = \prod_{m \in j} \delta_{\hat{a}_m \, a_m},
\end{equation}
that is, $q_{\{ a_m \}_{m \in \joint}}(\{\hat{a}_m \}_{m\in C})$ is 1 if $\hat{a}_m=a_m$ for all $m\in C$ and 0 otherwise.  To compare to eq.(\ref{eq:clas}), here $\{ a_m \}_{m \in \JM}$ corresponds to $\lambda$, $\cP_{NC}$ is the ``weight'' and $q$ is the deterministic model. Equation (\ref{e:nc_joint_deltas}) can then be easily rewritten as (\ref{e:nc_joint}), given in the main text.  This establishes the connection between the two ways of expressing classicality/non-contextuality, for scenarios in the hypergraph approach and marginal scenarios.

\section{Bell scenarios, their hypergraph-approach version and their macroscopic experiments}\label{se:corr-meas}

Mathematically, the MNC condition is a fairly straightforward generalisation of the Macroscopic Locality condition for nonlocality to general contextuality scenarios.  But physically the gedankenexperiment used to motivate it concerned a single system passing through a ``beam splitter,'' and this is different from the motivation given in \cite{ML}.  On examination this opens a number of issues of physical motivation, many of which will be relevant for other applications of physical principles to contextuality scenarios.

One of the main claims above was that, when specialised to Bell scenarios, the MNC principle constrains probabilistic models to the almost quantum set.  Below, the way in which Bell scenarios are represented in the hypergraph approach is briefly reviewed, and we present a way to interpret the macroscopic version of a Bell scenario in the hypergraph approach.

\subsection{Bell scenarios}

A typical Bell-type experiment consists of $n$ separated parties which have access each to a physical system.   The ``local'' measurements carried out by these parties are arranged so that they define spacelike separated events. In each run of the experiment, each party can subject their local system to their choice of one of $m$ local measurements, each with $d$ possible outcomes. The measurement choices are usually denoted by $x_k$ and the measurements outcomes by $a_k$, where $k$ labels the parties. Such a Bell scenario is thus characterised by the numbers $(n,m,d)$.  If the parties take note of the outcomes in each run of the experiment and gather statistics, they will eventually obtain a conditional probability distribution $P(a_1 \ldots a_n | x_1 \ldots x_n)$ (also referred to as \textit{correlations}).  Usually only probability distributions that obey the well-known ``no-signalling'' principle are of interest, meaning that marginalising over a local outcome $a_i$ will give a probability distribution that is independent of the local measurement setting $x_i$.  As well as the ``simultaneous'' measurements defined by a choice of measurement for each party, ``correlated measurements'' can be defined, which are important for the representation of Bell scenarios in the ALFS formalism.

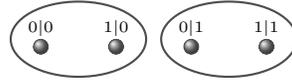
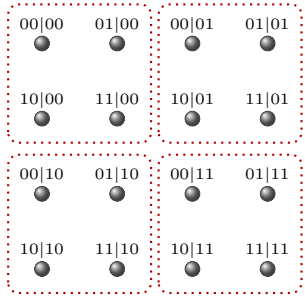
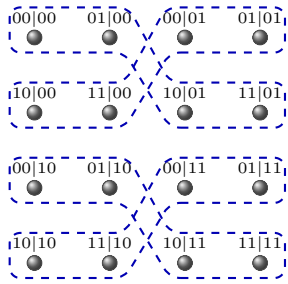
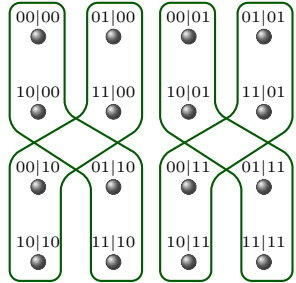
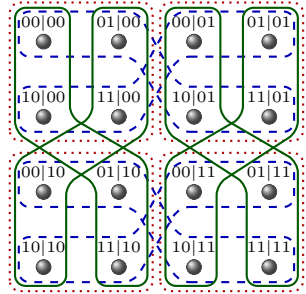
\begin{figure}
\definecolor{darkgreen}{rgb}{0,.5,0}
\begin{center}
\subfigure[Alice's two binary measurements $B_{1,2,2}$.]{
\label{CHSHa}
\begin{tikzpicture}
\clip (0,.5) rectangle (5,5);
\foreach \x in {0,1} \foreach \a in {0,1}
{
	\node[draw=gray!60!black,shape=circle,shading=ball, ball color=gray,scale=.5] (e) at (1,4-\a-2*\x) {} ;
	\node[above=0pt] at (e) {\tiny{$\a|\x$}} ;
}
\draw[thick,color=gray!60!black] (1,1.6) ellipse (.5cm and .9cm) ;
\draw[thick,color=gray!60!black] (1,3.6) ellipse (.5cm and .9cm) ;
\end{tikzpicture}}\hspace{.5cm}
\subfigure[Bob's two binary measurements $B_{1,2,2}$.]{
\label{CHSHb}
\begin{tikzpicture}
\clip (0,.5) rectangle (5,5);
\foreach \y in {0,1} \foreach \b in {0,1}
{
	\node[draw=gray!60!black,shape=circle,shading=ball, ball color=gray,scale=.5] (e) at (\b+2*\y+1,4) {} ;
	\node[above=0pt] at (e) {\tiny{$\b|\y$}} ;
}
\draw[thick,color=gray!60!black] (1.5,4.1) ellipse (.9cm and .5cm) ;
\draw[thick,color=gray!60!black] (3.5,4.1) ellipse (.9cm and .5cm) ;
\end{tikzpicture}}\hspace{.5cm}
\subfigure[Simultaneous measurements.]{
\label{sim_meas}
\hspace{.5cm}
\begin{tikzpicture}
\clip (0,.5) rectangle (5,5);
\foreach \x in {0,1} \foreach \y in {0,1} \foreach \a in {0,1} \foreach \b in {0,1}
{
	\node[draw=gray!60!black,shape=circle,shading=ball, ball color=gray,scale=.5] (e) at (\b+2*\y+1,4-\a-2*\x) {} ;
	\node[above=0pt] at (e) {\tiny{$\a\b|\x\y$}} ;
}
\foreach \x in {0,2} \foreach \y in {0,2} \draw[rounded corners,thick,color=red!70!black, dotted] (\x+0.55,\y+0.68) rectangle (\x+2.45,\y+2.52) ;
\end{tikzpicture}\hspace{.5cm}}\hspace{.5cm}
\subfigure[Bob's measurement choice depends on Alice's outcome.]{\hspace{.5cm}
\label{CHSHab}
\begin{tikzpicture}
\clip (0,.5) rectangle (5,5);
\foreach \x in {0,1} \foreach \y in {0,1} \foreach \a in {0,1} \foreach \b in {0,1}
{
	\node[draw=gray!60!black,shape=circle,shading=ball, ball color=gray,scale=.5] (e) at (\b+2*\y+1,4-\a-2*\x) {} ;
	\node[above=0pt] at (e) {\tiny{$\a\b|\x\y$}} ;
}
\foreach \y in {0,2} \draw[thick,color=blue!70!black,dashed,rounded corners] (2.8,\y+0.8) -- (4.33,\y+0.8) -- (4.33,\y+1.4) -- (2.8,\y+1.4) -- (2.25,\y+2.4) -- (0.67,\y+2.4) -- (0.67,\y+1.8) -- (2.2,\y+1.8) -- cycle ;
\foreach \y in {0,2} \draw[thick,color=blue!70!black,dashed,rounded corners] (2.8,\y+1.8) -- (4.33,\y+1.8) -- (4.33,\y+2.4) -- (2.75,\y+2.4) -- (2.2,\y+1.4) -- (0.67,\y+1.4) -- (0.67,\y+0.8) -- (2.2,\y+0.8) -- cycle ;
\end{tikzpicture}\hspace{.5cm}}\hspace{.5cm}
\subfigure[Alice's measurement choice depends on Bob's outcome.]{\hspace{.5cm}
\label{CHSHba}
\begin{tikzpicture}
\clip (0,.5) rectangle (5,5);
\foreach \x in {0,1} \foreach \y in {0,1} \foreach \a in {0,1} \foreach \b in {0,1}
{
	\node[draw=gray!60!black,shape=circle,shading=ball, ball color=gray,scale=.5] (e) at (\b+2*\y+1,4-\a-2*\x) {} ;
	\node[above=0pt] at (e) {\tiny{$\a\b|\x\y$}} ;
}
\foreach \x in {0,2} \draw[thick,color=darkgreen!70!black, rounded corners] (\x+0.62,2.4) -- (\x+0.62,0.75) -- (\x+1.3,0.75) -- (\x+1.3,2.1) -- (\x+2.38,2.8) -- (\x+2.38,4.45) -- (\x+1.65,4.45) -- (\x+1.65,3) -- cycle ;
\foreach \x in {0,2} \draw[thick,color=darkgreen!70!black, rounded corners] (5-\x-0.62,2.4) -- (5-\x-0.62,0.75) -- (5-\x-1.3,0.75) -- (5-\x-1.3,2.1) -- (5-\x-2.38,2.8) -- (5-\x-2.38,4.45) -- (5-\x-1.65,4.45) -- (5-\x-1.65,3) -- cycle ;
\end{tikzpicture}\hspace{.5cm}}
\subfigure[The full CHSH scenario $B_{2,2,2}$.]{
\label{CHSHc}
\begin{tikzpicture}
\clip (0,.5) rectangle (5,5);
\foreach \x in {0,1} \foreach \y in {0,1} \foreach \a in {0,1} \foreach \b in {0,1}
{
	\node[draw=gray!60!black,shape=circle,shading=ball, ball color=gray,scale=.5] (e) at (\b+2*\y+1,4-\a-2*\x) {} ;
	\node[above=0pt] at (e) {\tiny{$\a\b|\x\y$}} ;
}
\foreach \x in {0,2} \foreach \y in {0,2} \draw[rounded corners, thick, color=red!70!black, dotted] (\x+0.55,\y+0.68) rectangle (\x+2.45,\y+2.52) ;
%\foreach \x in {0,2} \foreach \y in {0,2} \draw[rounded corners,thick,blue] (\x+0.57,\y+0.7) rectangle (\x+2.43,\y+2.5) ;
\foreach \y in {0,2} \draw[thick,color=blue!70!black, dashed,rounded corners] (2.8,\y+0.8) -- (4.33,\y+0.8) -- (4.33,\y+1.4) -- (2.8,\y+1.4) -- (2.25,\y+2.4) -- (0.67,\y+2.4) -- (0.67,\y+1.8) -- (2.2,\y+1.8) -- cycle ;
\foreach \y in {0,2} \draw[thick,color=blue!70!black, dashed,rounded corners] (2.8,\y+1.8) -- (4.33,\y+1.8) -- (4.33,\y+2.4) -- (2.75,\y+2.4) -- (2.2,\y+1.4) -- (0.67,\y+1.4) -- (0.67,\y+0.8) -- (2.2,\y+0.8) -- cycle ;
\foreach \x in {0,2} \draw[thick,color=darkgreen!70!black, rounded corners] (\x+0.62,2.4) -- (\x+0.62,0.75) -- (\x+1.3,0.75) -- (\x+1.3,2.1) -- (\x+2.38,2.8) -- (\x+2.38,4.45) -- (\x+1.65,4.45) -- (\x+1.65,3) -- cycle ;
\foreach \x in {0,2} \draw[thick,color=darkgreen!70!black, rounded corners] (5-\x-0.62,2.4) -- (5-\x-0.62,0.75) -- (5-\x-1.3,0.75) -- (5-\x-1.3,2.1) -- (5-\x-2.38,2.8) -- (5-\x-2.38,4.45) -- (5-\x-1.65,4.45) -- (5-\x-1.65,3) -- cycle ;
\end{tikzpicture}}
\end{center}
\caption{Construction of the CHSH scenario $B_{2,2,2}$. This figure is originally Fig. 7 in \cite{AFLS}.}
\label{CHSH_sce}
\end{figure}

A correlated measurement in a Bell scenario is defined as follows. One party performs a local measurement $x_{i_1}$ and obtains an outcome $a_{i_1}$ which is communicated to the remaining parties. The second party in the protocol chooses a measurement $x_{i_2}$, which may depend on $a_{i_1}$. An outcome $a_{i_2}$ is obtained and communicated to the remaining parties. The protocol proceeds similarly for all parties, so that each party's measurement may depend on the previous parties' outcomes. The order in which the parties measure may also be defined dynamically throughout the protocol (see \cite{AFLS}, Def. 3.3.4).  This kind of protocol is often referred to as an example of a ``wiring protocol'' between parties.  These measurements are included in the hypergraph approach definition of Bell scenarios, which are denoted $B_{n,m,d}$.  See figure \ref{CHSH_sce} for the example of $B_{2,2,2}$, also known as the CHSH scenario\footnote{\label{f:bell-marginal}In the hypergraph approach, Bell scenarios are a special case of the marginal scenarios discussed in the previous appendix.  There are $n$ ``local'' sets of observables each containing $m$ observables, with $d$ outcomes each, such that the maximal jointly measurable sets are all sets composed of one observable from each of the local sets \cite{AFLS}.  In this case, the set of all measurement protocols defined in the previous appendix is the same as the set of all simultaneous and correlated measurements discussed here \cite{AFLS}.}.

As commented in \cite{AFLS}, the motivation for including the correlated measurements is mainly mathematical.  The simple hypergraph based framework allows the application of powerful graph-theoretic methods, and so it is useful to be able to treat Bell scenarios as a special case.  Most commonly in discussions of Bell scenarios, it is directly imposed that a \textit{local} measurement outcome at the $A$ wing given a setting in the $B$ wing should be indentifed with the same local outcome at $A$ given a different setting at $B$ (and similarly with $A$ and $B$ reversed).  This, without further restrictions, implies no-signalling, and is essentially what is done in \cite{ML} for instance.  But this is not directly representable in the hypergraph approach, which deals directly with global outcomes and only allows these to be indentified with each other.  To impose the no-signalling principle directly on these scenarios would require either an ad-hoc restriction or a more complicated general formalism (\textit{e.g.}~allowing the identification of sets of global outcomes as well as individual outcomes).  Adding correlated measurements circumvents this problem because, when they are present, only probabilistic models that satisfy the no-signalling principle are consistent. Similarly, their inclusion ensures that the definition of quantum models for general contextuality scenarios specialises to the usual definition of quantum correlations for Bell scenarios, which is as follows.

\begin{defn}\label{qcorrvan}\textbf{Quantum Correlations}\\
Let $(n,m,d)$ be a Bell Scenario. A conditional probability distribution $P(a_1 \ldots a_n | x_1 \ldots x_n)$ is quantum if there exists a Hilbert space $\mathcal{H}$, a state $\rho\in\mathcal{B}_{+,1}(\mathcal{H})$ and $m$ projective measurements $\{P^k_{a_k|x_k}\}_{a_k=1 \ldots d}$ for each party $k$ such that:
\begin{enumerate}
\item $\sum_{a_k=1}^d P^k_{a_k|x_k} = \mathbbm{1}_{\mathcal{H}}$ for all $x_k = 1 \ldots m$,
\item $[P^k_{a_k|x_k},P^{k^\prime}_{a_{k^\prime}|x_{k^\prime}}]=0$ for all $k \neq {k^\prime}$,
\item $P(a_1 \ldots a_n | x_1 \ldots x_n) = \mathrm{tr} (P^1_{a_1|x_1} \, \ldots \, P^n_{a_n|x_n} \, \rho)$.
\end{enumerate}
\end{defn}

Similarly non-contextuality, when applied to Bell scenarios, is equivalent to locality.

Thus, correlated measurements play a useful role mathematically.  Physically, however, if the scenario is meant to represent a choice of measurement for $n$ spacelike-separated parties, the protocol for correlated measurements cannot actually be carried out, and so we seem to have a contradiction between the most important application of Bell scenarios and the inclusion of correlated measurements. It is interesting to consider how this affects the motivation of principles applied to Bell scenarios in general.  Below, we only consider motivations for applying the MNC principle to the results of Bell experiments.

\subsection{Bell scenarios and the MNC principle}

As mentioned in the main text, applying the MNC condition to bipartite Bell scenarios results in a condition that is stronger than the original Macroscopic Locality condition \cite{ML}, which applies exclusively to this type of scenario.  However, the motivations for imposing the two conditions in this case differ substantially.

In section \ref{se:MNC} we gave a physical picture to motivate the MNC condition, of a single system passing through a measurement device after which it ends up hitting one of many detectors; the macroscopic version of the experiment simply consisted of many systems passing through a similar apparatus.  For physical experiments of this form, the motivation for applying the MNC condition has the most clarity. Some such gedankenexperiments do indeed correspond to the $B_{n,m,d}$ scenarios discussed above, and in this sense these motivating comments are valid for these scenarios.  However, this easy answer has little to do with the special status of Bell scenarios; because it takes place at one location in spacetime, such an experiment cannot (directly, at least) invoke any motivations stemming from locality or relativity.  It is the standard Bell experiment, involving spacelike separated parties, that is the physical experiment of interest which, to a large extent, motivates the study of these scenarios in the first place.

In the original argument for ML, the gedankenexperiment discussed corresponds to this standard case: one run of the experiment involves a pair of particles, each passing through a measurement device at distant locations.  This is important, because the separation of the two measurement devices is what motivates the application of the no-signalling and locality conditions used in the argument.  The main strength of this kind of motivation, in contrast to that for the MNC condition, is that it does not depend on any assumptions about the experimental protocols under consideration apart from the separation of the parties (``device independence''). 

It is difficult to directly extend this sort of motivation to the MNC condition applied to Bell scenarios.  To begin with, MNC concerns ``intensities'' corresponding to the \textit{global} measurement outcomes, but in the ML gedankenexperiment only the counts of local outcomes are available, not counts of how many \textit{pairs} of particles hit a particular \textit{pair} of detectors.  Secondly, it is not immediately clear how to motivate the inclusion of correlated measurements. Thus it is not clear if MNC implies any constraints on experimental results for the spacelike-separated version of the Bell experiment\footnote{One might wonder if the global intensities and the inclusion of correlated measurements could be discarded without affected the strength of the MNC condition.  But these actually constitute the difference between the ML and MNC conditions, and these must therefore account for the increased strength of the MNC condition over the ML condition.   For example, from the constraint in equation (\ref{e:central_lim}) in the main text, when we have two outcomes that are alternatives in one measurement, the corresponding entry in the covariance matrix is determined.  Thus, the more measurements are given, the more constraints there are on this matrix, and so considering the correlated measurements changes the constraints.}.

Given a single system experiment, one the other hand, these problems are avoided.  To apply MNC to the ``true'' Bell experiment, therefore, we need to relate this experiment to a single system experiment.  To examine this issue, let us consider the relatively familiar case of quantum theory.

Let us first consider the usual Bell experiment, which will be called ``experiment 1'', with two separated parties.  The experimental results are a conditional probability distribution, and the mathematical model of experiment 1 is as in Def. \ref{qcorrvan} given in the previous section, \textit{i.e.}~a Hilbert space, state, and projectors with certain properties.  In quantum mechanics, it is always (in principle) possible to set up a single system experiment that can be described by the same quantum model, and thus has corresponding experimental results.  Call this ``experiment 2''.  Furthermore, in experiment 2 it is (in principle) possible to add measurements corresponding to all the correlated measurements for the appropriate Bell scenario, making experiment 2 a realisation of a quantum model on a Bell scenario in the hypergraph approach.  Thus, in quantum theory at least, there is a strong relation between any given ``true'' Bell experiment and some single system experiment which can be described by a Bell scenario in the hypergraph approach.

To apply MNC to the separated Bell experiment in general, we need this property to remain true in whatever theoretical framework is being applied.  That is, for any possible experimental results for a given Bell experiment, there should exist (in principle) a single system experiment that can be described by the corresponding Bell scenario in the hypergraph approach, and which gives the corresponding experimental results.  If this was true, then any constraint implied by a principle for the single-system experiment would also apply to the separated Bell scenario.  In this case MNC would indeed apply to the separated Bell experiment and the almost quantum bound would be respected.

This is a rather abstract assumption on the relation of theory to experiment.  It does provide at least one way to apply MNC to experiments with separated parties, even if the motivation is not as simple as merely invoking relativistic causality.  Furthermore it might be agued that, whenever nonlocality is considered as a subcategory of contextuality for all intents and purposes, some assumptions of this nature are always implicitly made.

\small
\bibliographystyle{unsrt}
\bibliography{mnc}

\end{document}